\documentclass[a4paper,UKenglish,cleveref,thm-restate,table]{lipics-v2021}

\usepackage[subrefformat=simple,labelformat=simple,position=b]{subcaption}

\usepackage{nicefrac}
\usepackage{mathtools}
\usepackage{xspace}
\usepackage{complexity}
\usepackage{array}
\usepackage{multirow}
\usepackage{hhline}
\usepackage{nicefrac}

\usepackage{thmtools}

\usepackage{todonotes}
\usepackage{xcolor}
\usepackage[linesnumbered,ruled,commentsnumbered,longend]{algorithm2e}

\crefname{figure}{Figure}{Figures}
\crefname{theorem}{Theorem}{Theorems}
\crefname{lemma}{Lemma}{Lemmas}
\crefname{corollary}{Corollary}{Corollaries}
\crefname{section}{Section}{Sections}
\crefname{appendix}{Appendix}{Appendices}
\crefname{remark}{Remark}{Remarks}
\crefname{claim}{Claim}{Claims}
\crefname{conjecture}{Conjecture}{Conjectures}
\crefname{observation}{Observation}{Observations}

\SetKwInput{KwData}{Input}
\SetKwInput{KwResult}{Output}

\newcommand{\matin}{M} 
\newcommand{\mattar}{{M'}} 
\newcommand{\exposed}{isolated\xspace}
\newcommand{\ps}{S} 

\DeclareMathOperator{\oH}{\overrightarrow{H}}
\newcommand{\BigO}{\mathcal{O}}

\definecolor{tableblue}{rgb}{0.1,0.1,0.7}
\definecolor{tablecell}{rgb}{0.1,0.7,0.7}

\graphicspath{{./figures/}}

\title{Flipping odd matchings in geometric and combinatorial settings}
\titlerunning{Flipping odd matchings in geometric and combinatorial settings}

\author{Oswin Aichholzer}{Graz University of Technology}{oswin.aichholzer@tugraz.at}{https://orcid.org/0000-0002-2364-0583}{}
\author{Sofia Brenner}{University of Kassel}{sbrenner@mathematik.uni-kassel.de}{https://orcid.org/0009-0006-8512-2569}{German Research Foundation (DFG) grant 522790373}
\author{Joseph Dorfer}{Graz University of Technology}{joseph.dorfer@tugraz.at}{https://orcid.org/0009-0004-9276-7870}{Austrian Science Foundation (FWF) grant 10.55776/DOC183}
\author{Hung P. Hoang}{TU Wien}{phoang@ac.tuwien.ac.at}{https://orcid.org/0000-0001-7883-4134}{Austrian Science Foundation (FWF) project Y 1329 START-Programm}
\author{Daniel Perz}{University of Perugia}{daniel.perz@unipg.it}{https://orcid.org/0000-0002-6557-2355}{MUR PRIN project 2022ME9Z78 - ``NextGRAAL:~Next-generation algorithms for constrained GRAph visuALization''}
\author{Christian Rieck}{University of Kassel}{christian.rieck@mathematik.uni-kassel.de}{https://orcid.org/0000-0003-0846-5163}{German Research Foundation (DFG) grant 522790373}
\author{Francesco Verciani}{University of Kassel}{francesco.verciani@mathematik.uni-kassel.de}{https://orcid.org/0009-0000-0536-8646}{German Research Foundation (DFG) grant 522790373}

\authorrunning{O.~Aichholzer, S.~Brenner, J.~Dorfer, H.~Hoang, D.~Perz, C.~Rieck, and F.~Verciani}

\Copyright{Oswin Aichholzer, Sofia Brenner, Joseph Dorfer, Hung P.~Hoang, Daniel Perz, Christian Rieck, and Francesco Verciani}



\nolinenumbers

\hideLIPIcs

\EventEditors{Vida Dujmovi\'c and Fabrizio Montecchiani}
\EventNoEds{2}
\EventLongTitle{33rd International Symposium on Graph Drawing and Network Visualization (GD 2025)}
\EventShortTitle{GD 2025}
\EventAcronym{GD}
\EventYear{2025}
\EventDate{September 24--26, 2025}
\EventLocation{Norrk\"{o}ping, Sweden}
\EventLogo{}
\SeriesVolume{357}
\ArticleNo{10}

\begin{document}

    \maketitle

    \begin{abstract}
        We study the problem of reconfiguring \emph{odd matchings}, that is, matchings that cover all but a single vertex. Our reconfiguration operation is a so-called \emph{flip} where the unmatched vertex of the first matching gets matched, while consequently another vertex becomes unmatched.
        We consider two distinct settings: the \emph{geometric setting}, in which the vertices are points embedded in the plane and all occurring odd matchings are crossing-free, and a \emph{combinatorial setting}, in which we consider odd matchings in general graphs.
        
        For the latter setting, we provide a complete polynomial time checkable characterization of graphs in which any two odd matchings can be reconfigured into each another. 
        This complements the previously known result that the flip graph is always connected in the geometric setting~\cite{oddmatchings}. 
        In the combinatorial setting, we prove that the diameter of the flip graph, if connected, is linear in the number of vertices. Furthermore, we establish that deciding whether there exists a flip sequence of length $k$ transforming one given matching into another is \NP-complete in both the combinatorial and the geometric settings. To prove the latter, we introduce a framework that allows us to transform partial order types into general position with only polynomial overhead. 
        Finally, we demonstrate that when parameterized by the flip distance $k$, the problem is fixed-parameter tractable (\FPT) in the geometric setting when restricted to convex point sets.
    \end{abstract}
    \pagebreak
    \section{Introduction}
\label{sec:introduction}

In many areas of discrete mathematics and theoretical computer science, one is interested not only in finding individual solutions to combinatorial problems, but also in understanding how these solutions relate to each other---specifically, whether it is possible to gradually transform one solution into another through a sequence of small, valid changes (so-called~\emph{flips}), while maintaining feasibility at each step.
This framework is known as \emph{reconfiguration}. 
It~provides insight into the structure of the solution space and has applications in various areas. 
For example in optimization~\cite{lin1973effective}, as every optimization problem can be reformulated as a reconfiguration problem by introducing an appropriate threshold on the objective functions~\cite{ItoDHPSUU11}, as well as combinatorial enumeration and robustness~\cite{avis1996reverse,mütze2024combinatorialgraycodesanupdated}.
Moreover,~many classical results in combinatorics naturally fall within this framework, e.g., the 5-color theorem~\cite{MR30735} or Vizing's theorem~\cite{assadi2024vizingstheoremnearlineartime,vizing1965critical}.
We also refer to the following surveys~\cite{Nishimura18,Heuvel13}.
A~prominent example is the reconfiguration of triangulations of a convex $n$-gon; since each binary tree corresponds to the geometric dual of such a triangulation, this problem is equivalent to performing rotations on binary trees~\cite{HURTADO1999179,MR928904}.

This leads to structural and algorithmic questions: 
(1) Can any structure of a certain class be transformed into any other using the given flip operation?
(2) What is the worst-case number of flips required to perform such a transformation?
(3) Given two specific structures, how can we compute the minimum number of flips needed to reconfigure one into the other?
These questions can also be rephrased in terms of the flip graph, where vertices represent the structures and edges correspond to flips: 
(1) Is the flip graph \emph{connected}? 
(2) What is the \emph{diameter} of the flip graph? 
(3) What is the computational \emph{complexity} of finding shortest paths between vertices in the flip graph?

A central example is the reconfiguration of \emph{matchings} in graphs, particularly \emph{perfect matchings}, which are sets of pairwise non-adjacent edges covering every vertex.
In this setting, a small local change corresponds to removing a set of edges and inserting another along an alternating cycle, while maintaining a valid matching throughout. 

We study the reconfiguration of \emph{odd matchings}, which are matchings that cover all but one vertex in a graph.
This variant arises naturally in graphs without perfect matchings due to parity constraints, or in point sets of odd cardinality.
In such cases, we refer to the graph or point set as being of odd \emph{order}, where the order is the number of vertices or points.
Each flip adds an edge between the unmatched vertex (which we also refer to as \emph{\exposed}) and another vertex $v$ and removes the matching edge $vw$ that was previously incident to~$v$. 
This~way, $w$ becomes the new \exposed vertex.
We consider two distinct settings:
In the \emph{geometric setting}, the vertices are points in the plane, and matchings consist of straight lines between points that do not share endpoints and that do not cross.
Each such \emph{geometric flip} yields a (plane) odd matching.
In the \emph{combinatorial setting}, we consider odd matchings in an input graph $G$, and a \emph{combinatorial flip} consists of removing and adding an edge of $G$ such that the result is again an odd matching.
For an illustration of both settings, see~\cref{fig:different-settings}.

\paragraph*{Problem description} 
We study the problem \textsc{Flipping Odd Matchings}, which concerns reconfiguring one odd matching into another through a sequence of valid flip operations. 
We examine the problem in both the geometric and the combinatorial setting with their corresponding flip operations. 
The associated decision problem asks whether, given two odd matchings $\matin$ and $\mattar$ and an integer $k \in \mathbb{N}$, there exists a sequence of at most $k$ flip operations that transforms $\matin$ into $\mattar$ within the respective setting.

\begin{figure}[htb]
    \centering
    \includegraphics[]{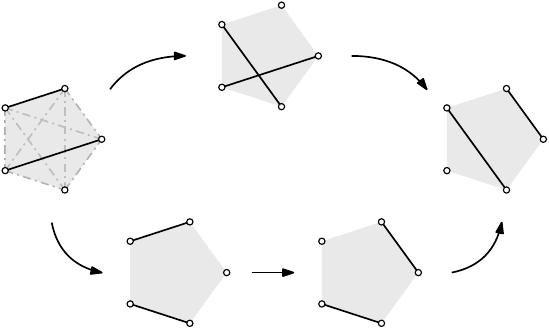}
    \caption{This figure illustrates the two distinct settings and their corresponding flip types.
    The odd matching on the left must be reconfigured into the odd matching on the right.
    The clockwise flip sequence represents a combinatorial reconfiguration, while the counterclockwise sequence is geometric.
    Dotted-dashed lines in the initial matching indicate additional edges present in the graph.}
    \label{fig:different-settings}
\end{figure}

\paragraph*{Our contribution}
In this paper, we present a number of novel results on reconfiguring odd matchings in geometric and combinatorial settings.
In particular, we provide a complete characterization of graphs in which each two odd matchings can be flipped into each other.
Moreover, if the corresponding (combinatorial) flip graph is connected, its diameter is linear in the number of vertices of the input graph.
Furthermore, we show that the decision variant of \textsc{Flipping Odd Matchings} is \NP-complete in both considered settings, but fixed-parameter tractable in the case of convex point sets when parameterized in the flip distance.

A detailed summary of our results is given in~\cref{tab:results}.

\begin{table}[htb]
\caption{Current status of reconfiguring odd matchings; new results are highlighted in light blue.}
\centering
\begin{tabular}{c >{\raggedright}p{3.3cm} p{2.7cm} p{3.2cm} p{2.25cm} }
	& \sffamily\textbf{Setting} & \sffamily\textbf{Connectedness} & \sffamily\textbf{Diameter} & \sffamily\textbf{Complexity}
	\\\hhline{-----}
    \parbox[t]{3mm}{\multirow{5}{*}{\rotatebox[origin=c]{90}{\sffamily\textbf{Geometric}\hspace*{0.1cm}}}} & \multirow{3}{*}{Convex} & \multirow{3}{*}{Yes~\cite{oddmatchings}} & \cellcolor{tablecell!25} & \cellcolor{tablecell!25}
    \\
     & & & \multirow{-2}{*}{\cellcolor{tablecell!25}\shortstack[l]{at most $\nicefrac{3n}{2}-\BigO(1)$,\\ \textcolor{tableblue}{\cref{cor:radius}}}} & \cellcolor{tablecell!25}\\
    & & & at least $n-2$~\cite{oddmatchings} & \multirow{-3}{*}{\cellcolor{tablecell!25}\shortstack[l]{\FPT,\\ \textcolor{tableblue}{\cref{thm:FPT}}}}\\
    \hhline{~|----}
	& \multirow{2}{*}{General position} & \multirow{2}{*}{Yes~\cite{oddmatchings}} & \multirow{2}{*}{$\Omega(n)$ and $\BigO(n^2)$~\cite{oddmatchings}} & \cellcolor{tablecell!25}\\
    & & & & \multirow{-2}{*}{\cellcolor{tablecell!25}\shortstack[l]{\NP-complete,\\ \textcolor{tableblue}{\cref{thm:general_hard}}}}\\
    \hhline{-----}
	\parbox[t]{3mm}{\multirow{6}{*}{\rotatebox[origin=c]{90}{ \sffamily\textbf{Combinatorial}\hspace*{0.2cm}}}} & \multirow{2}{*}{Rectangular grids} & \cellcolor{tablecell!25} & \cellcolor{tablecell!25} & \cellcolor{tablecell!25}\\
    & & \multirow{-2}{*}{\cellcolor{tablecell!25}\shortstack[l]{Yes,\\ \textcolor{tableblue}{\cref{rem:grid}}}} & \multirow{-2}{*}{\cellcolor{tablecell!25}\shortstack[l]{$\Theta(n)$,\\ \textcolor{tableblue}{\cref{thm:linear}}}} & \multirow{-2}{*}{\cellcolor{tablecell!25}\shortstack[l]{\NP-complete,\\ \textcolor{tableblue}{\cref{thm:grid_hard}}}}\\
    \hhline{~|----}
	& \multirow{2}{*}{Planar graphs} & \cellcolor{tablecell!25} & \cellcolor{tablecell!25} & \cellcolor{tablecell!25} \\
    & & \multirow{-2}{*}{\cellcolor{tablecell!25}\shortstack[l]{Characterization,\\ \textcolor{tableblue}{\cref{thm:reconfigurable}}}} & \multirow{-2}{*}{\cellcolor{tablecell!25}\shortstack[l]{if connected: $\Theta(n)$,\\ \textcolor{tableblue}{\cref{thm:linear}}}} & \multirow{-2}{*}{\cellcolor{tablecell!25}\shortstack[l]{\NP-complete,\\ \textcolor{tableblue}{\cref{thm:grid_hard}}}}\\
    \hhline{~|----}
	& \multirow{2}{*}{General graphs} & \cellcolor{tablecell!25} & \cellcolor{tablecell!25} & \cellcolor{tablecell!25}\\
    & & \multirow{-2}{*}{\cellcolor{tablecell!25}\shortstack[l]{Characterization,\\ \textcolor{tableblue}{\cref{thm:reconfigurable}}}} & \multirow{-2}{*}{\cellcolor{tablecell!25}\shortstack[l]{if connected: $\Theta(n)$,\\ \textcolor{tableblue}{\cref{thm:linear}}}} & \multirow{-2}{*}{\cellcolor{tablecell!25}\shortstack[l]{\NP-complete,\\ \textcolor{tableblue}{\cref{thm:grid_hard}}}}
\end{tabular}
\label{tab:results}
\end{table}

All technical details and proofs for statements marked with $(\star)$ are given in the appendix.
\pagebreak
\paragraph*{Related work}
In the geometric setting, it was recently shown in~\cite{oddmatchings} that for a set of $n$ points in general position in the plane, the flip graph of odd matchings is connected, with diameter bounded between $\Omega(n)$ and $\BigO(n^2)$. 
In contrast, for perfect matchings in the geometric setting, the connectedness of the flip graph remains an open question when flips are restricted to alternating cycles of sublinear length. 
However, if no such restriction is imposed, the flip graph is known to be connected~\cite{articlematchings_2}. 
Recent work shows that finding shortest flip sequences is \NP-hard for point sets in general position, when the flip size is restricted to 4-cycles~\cite{binucci2025flippingmatchingshard}.

In the combinatorial setting, the connectivity of the flip graph for perfect matchings has been studied in~\cite{bonamy2019perfectmatchingreconfigurationproblem}.
The authors showed that the problem is \PSPACE-complete for several graph classes, including split graphs and bipartite graphs of bounded bandwidth with maximum degree five. 
In contrast, they proved that it can be solved in polynomial time for strongly orderable graphs, outerplanar graphs, and cographs. 
In these positive cases, they also provided flip sequences of linear length.

Odd matchings have also been studied in the combinatorial setting, particularly in the special case of the sliding block puzzle known as \emph{Gourds}~\cite{hamersma2020gourds,kakimura2024reconfiguration}. 
Here, the underlying graph is a triangular grid graph, where matching edges are assigned distinct colors, and in each flip, the added edge inherits the color of the removed edge.
The connectedness of the flip graph is fully characterized when the underlying graph is a solid triangular grid graph~\cite{hamersma2020gourds}.
For~triangular grid graphs with holes, sufficient conditions for connectedness are known, but a complete characterization remains open~\cite{kakimura2024reconfiguration}.

We refer to~\cite{BOSE200960} for a survey of additional results and open questions on the reconfiguration of planar graphs in both the combinatorial and geometric settings.
Reconfiguration in general is also closely linked to \emph{Gray codes}~\cite{mütze2024combinatorialgraycodesanupdated}, which aim to list all feasible configurations in an order such that each configuration differs from the previous one only by a single flip. 
Such an ordering corresponds to identifying a Hamiltonian cycle in the flip graph, providing a compact and systematic traversal of the solution space~\cite{hhmmn-gtgpcp-99,articlematchings,HURTADO1999179,rivera2001hamilton}.
    \section{Preliminaries and basic observations}
\label{sec:prelim-obs}

\subparagraph*{Geometric and combinatorial settings.} 
We consider odd matchings in two different settings. 
In the \emph{geometric setting}, let $\ps$ be a set of $n$ points embedded in the plane.
A line segment between two points of $\ps$ is an \emph{edge}.
A \emph{matching} on $\ps$ is a set of edges whose endpoints are pairwise distinct, said to be \emph{plane} if no two edges cross.
An \emph{odd matching} $M$ on $\ps$ is a plane matching on $\ps$ where exactly one point in $\ps$ is not an endpoint of any edge of $M$. 
A~\emph{geometric flip} in an odd matching removes one edge and inserts another one so that the result remains an plane matching.
In essence, an odd matching on a point set $\ps$ is simply a plane odd matching on the complete graph with vertex set $\ps$.

In the \emph{combinatorial setting}, we consider odd matchings on a fixed given graph~$G$. 
A~\emph{matching} $M$ in a simple graph $G=(V,E)$ is a set of pairwise disjoint edges, that is, no two edges in $M$ share a common vertex. 
It is called \emph{odd} if exactly one vertex of $G$ is not incident to an edge in $M$.
A \emph{combinatorial flip} in an odd matching $M$ removes one edge and replaces it with another from $G$, preserving the property of being an odd matching.
In both settings, the new edge has to be incident to the \exposed vertex of the previous odd matching.

The \emph{flip graph} of a graph $G$ (or point set $\ps$) has as its vertices the odd matchings of~$G$ (resp., $\ps$), with edges connecting odd matchings that differ in a single flip.
We say that an odd matching $M$ can be \emph{reconfigured} to another odd matching $M'$ if there exists a flip sequence transforming $M$ into $M'$, i.e., there is a path between them in the flip graph. 
A~graph or point set is said to be \emph{reconfigurable} if its flip graph is connected, meaning any two odd matchings can be reconfigured into one another. 
The \emph{flip distance} $d(M, M')$ between two matchings $M$ and $M'$ is the length of a shortest path between them in the flip graph. 
The~\emph{diameter} of the flip graph is the maximum flip distance over all pairs of odd matchings. 
Its \emph{radius} is defined as $\min_{M} \max_{M'} d(M, M')$, and this minimum is attained by odd matchings that minimize their farthest distance to any other odd matching.
Note that the diameter is at least the radius and at most twice the radius.

\subparagraph*{Union and symmetric difference of matchings in both settings.} 
Let $M$ and $M'$ be odd matchings.
An \emph{$M$-alternating path} is a path that alternates edges in~$M$ and not in~$M$; $M$-alternating cycles are defined analogously.
Then $M \cup M'$ is a disjoint union of
\begin{itemize}
    \item a path of even length (possibly zero) that connects the \exposed vertex in $M$ to the \exposed vertex in $M'$ and alternates edges of $M'$ and $M$,
    \item cycles of even length alternating edges of $M$ and $M'$, and
    \item edges that lie in both $M$ and $M'$, which we call \emph{happy edges}.
\end{itemize}

The \emph{symmetric difference} $M \triangle M'$ is the set of edges that appear in only one of the two odd matchings.
In other words, it is obtained from $M \cup M'$ by removing all happy edges.

\subparagraph*{Flipping along alternating paths.} 
We consider the combinatorial setting.
Let $P$ be an \mbox{$M$-alternating} path from the \exposed vertex $v_M$ to some vertex $v$, where the last edge belongs to $M$. 
We can flip the matching along $P$, hence replacing $M \cap P$ by the edges in $P \setminus M$; the new \exposed vertex is $v$.  
Next, we consider $M$-alternating paths to cycles in $M \triangle M'$.

\begin{lemma}\label{lem:alternating_paths}
Let $M$ and $M'$ be odd matchings of a graph $G$. Suppose $M$ can be reconfigured to $M'$. For every cycle $C$ in $M \triangle M'$, there exists a vertex $v \in V(C)$ such that $v$ has an $M$-alternating path to the \exposed vertex $v_{M}$.%
\end{lemma}

\begin{proof}
Let $C$ be a cycle in $M \triangle M'$. 
The edges in $M \cap C$ are eventually flipped to reconfigure $M$ into $M'$. 
After flipping such an edge $e=vw$, one of its endpoints, say~$v$, becomes \exposed in the corresponding intermediate matching.
Consequently, $G-v$ has a perfect matching. 
Hence, $M \setminus \{e\}$~is not maximal in $G-v$. 
By~Berge's theorem~\cite{MR94811}, there exists an $M$-alternating path from $w$ to~$v_M$. 
\end{proof}

Consider a cycle~$C$ in $M \triangle M'$ and assume that there exists an $M$-alternating path~$P$ from $v_M$ to $v \in C$. 
Then $M \cap C$ can be replaced by $M' \cap C$ by flipping the edges along~$P$, around~$C$, and then restoring the edges on~$P$. 
After this, $v_M$ is again the \exposed vertex. 
When applying the described flipping procedure in the special case that $v_M$ is adjacent to a vertex $v \in C$, we say that we \emph{switch}~$C$.
If $C$ has $k$ edges, then it takes $k+1$ flips to switch~$C$. 
The additional flip occurs because we need to first place the \exposed vertex on the cycle. 
This also holds for the geometric setting if $C$ is planar; see~\cref{fig:mats} for an illustration.

\begin{figure}[htb]
	\centering
	\includegraphics[page=1]{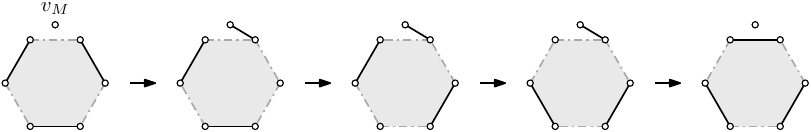}
	\caption{How to switch an alternating cycle.}
	\label{fig:mats}
\end{figure}

However, this does not hold in general in the geometric setting.
Suppose we wish to switch the cycle illustrated in~\cref{fig:non-plane-alternating-cycle}. 
The only feasible flips initially remove the blue edge ``closest'' to $v_M$. 
However, regardless of how this edge is removed, a red edge cannot be added immediately in the next flip, as doing so would cause it to cross a remaining blue edge.

\begin{figure}[htb]
	\centering
	\includegraphics[page=2]{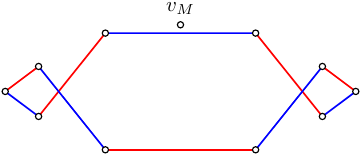}
	\caption{A non-plane alternating cycle in the geometric setting.}
	\label{fig:non-plane-alternating-cycle}
\end{figure}
    \section{Reachability}
\label{sec:flipgraph}

We consider the question of whether two odd matchings of a given graph can be reconfigured into each other.
For the geometric setting this has been proven to be always possible~\cite{oddmatchings}. 
In~contrast, the situation in the combinatorial setting is fundamentally different, as there exist simple examples where reconfiguration is not possible, as illustrated in~\cref{fig:counterexample-graph-reconfiguration}.

\begin{figure}[htb]
    \centering
    \includegraphics[]{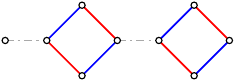}
    \caption{A connected graph with two odd matchings (red and blue, respectively) that cannot be reconfigured into each another. Specifically, the rightmost cycle cannot be reconfigured. This example can be extended for any odd $n$ by increasing the size of the rightmost cycle.}
    \label{fig:counterexample-graph-reconfiguration}
\end{figure}

In this section, we characterize the graphs that are reconfigurable, that is, for which the flip graph is connected.
Further, we give a simple polynomial-time algorithm to decide whether a graph is reconfigurable.

\begin{lemma}\label{lem:redundantedges}
Let $G$ be a graph of odd order, and consider an edge $e$.
\begin{enumerate}
    \item Suppose that $e$ is not contained in any odd matching of $G$. Then~$G$ is reconfigurable if and only if $G-e$ is reconfigurable. 
    \item Suppose that $e=uv$ is contained in every odd matching of~$G$. Then~$G$ is reconfigurable if and only if $G-u-v$ is reconfigurable.
\end{enumerate}
\end{lemma}

\begin{proof}
If $e$ is not contained in any odd matching of $G$, then clearly $G$ is reconfigurable if and only if $G-e$ is, as $e$ never appears in a flip sequence in $G$. 
Now suppose that $e=uv$ is contained in every odd matching of $G$.
As no reconfiguration sequence in $G$ ever flips $e$, it is easy to see that $G$ is reconfigurable if and only if $G-u-v$ is reconfigurable. 
\end{proof}

\begin{theorem}\label{thm:reconfigurable}
Let $G$ be a graph of odd order. Then $G$ is reconfigurable if and only if for every edge $e = uv \in E(G)$, one of the following holds: 
\begin{enumerate}
    \item the edge $e$ is either contained in all or in no odd matchings of $G$, or
    \item at least one of $G-u$ and $G-v$ contains a perfect matching.
\end{enumerate}
\end{theorem}

\begin{proof}
By~\cref{lem:redundantedges}, we may assume that every edge of $G$ is contained in some, but not all odd matchings of~$G$. 
First assume that for every edge $e = uv \in E(G)$, $G-u$ or $G-v$ contains a perfect matching. 
Let $M$ and $M'$ be odd matchings in~$G$ with \exposed vertices $v_M$ and $v_{M'}$, respectively.
Flipping edges along the path in $M \triangle M'$ maps $v_M$ to $v_{M'}$, so assume $v_M = v_{M'}$. 
Then $M \triangle M'$ is a collection of cycles. 
Let~$C$ be a cycle in $M \triangle M'$ and let $e = uv \in M$ be an edge on $C$. 
By assumption, $G-u$ or $G-v$ contains a perfect matching, say $G-u$. 
By applying Berge's theorem~\cite{MR94811} to $M \setminus \{e\}$ in $G-u$, there exists an $M$-alternating path from $v_M$ to $v$. 
Using this alternating path to switch $C$ yields a matching~$M_1$ with $|M_1 \triangle M'| < |M \triangle M'|$. 
Inductively, this yields a flip sequence from $M$ to~$M'$.

Now assume that $G$ contains an edge $e=uv$ such that $G-u$ and $G-v$ do not have a perfect matching.
In other words, $u$ and $v$ are covered in all odd matchings.
By assumption, there exists an odd matching~$M_1$ containing $e$ and an odd matching~$M_2$ not containing~$e$.
To reconfigure $M_1$ to $M_2$, we eventually need to perform a flip involving $e$.
However, this isolates $u$ or $v$, which is a contradiction.
Hence,~$G$ is not reconfigurable.
\end{proof}

\begin{remark}
Let $G$ be a factor-critical graph, i.e., $G-v$ has a perfect matching for every $v \in V(G)$. Clearly, $G$ then satisfies condition (2) in~\cref{thm:reconfigurable}, and hence $G$ is reconfigurable.
\end{remark}

\begin{restatable}[$\star$]{remark}{rectangularGridReconfigurable}\label{rem:grid}
    For any odd values $w,h\in \mathbb{N}$ the $w\times h$ square grid graph, i.e., the Cartesian product $P_w \Box P_h$, is reconfigurable.
\end{restatable}

There is a straightforward polynomial-time algorithm to check the conditions of the characterization (\cref{thm:reconfigurable}), i.e., deciding whether a graph is reconfigurable.

\begin{restatable}[$\star$]{theorem}{theoremCharacterizationAlgorithm}\label{thm:check-reconfigurability}
    Given a graph $G$, we can determine whether the flip graph of odd matchings in~$G$ is connected in $\BigO(n^{4.5})$ time. In other words, the reconfigurability of a graph can be checked in polynomial time.
\end{restatable}
    \section{Radius and diameter of the flip graph}
\label{sec:easy_setting}

We now study the diameter of the flip graphs in both combinatorial and geometric settings.

\subsection{Linear diameter of the flip graph in general graphs}
In this section, we show that the diameter of the flip graph in the combinatorial setting is linear in the order $n$ of the input graph. 
Given an input odd matching $M$ and a target matching $M'$ in a graph $G$ of order~$n$, such that it is possible to reconfigure $M$ to $M'$, the idea is to switch the cycles in $M \triangle M'$, traversing them in a tree-like fashion. In order to do so, we use an auxiliary tree structure. 
The proof proceeds in two steps. 
In the first part, we construct an auxiliary graph for the cycle traversal. 
In the second part, we prove that there exists a flip sequence of length linear in $n$ reconfiguring $M$ into $M'$. 
In the following, we always assume that the \exposed vertices of $M$ and $M'$ coincide.
This will later be achieved by first flipping the path in $M \triangle M'$ between these two vertices.

\paragraph*{Construction of the auxiliary graph}
We construct an auxiliary directed graph $\oH$ that will be used to construct the flip sequence. 
At any step of the construction, let $H$ denote the underlying undirected graph of~$\oH$. 

We initialize $\oH$ by introducing, for each vertex $v$ of $G$, two corresponding vertices in~$\oH$: an \emph{in-vertex} $v_{\text{in}}$ and an \emph{out-vertex} $v_{\text{out}}$; for an illustration see~\cref{fig:linear-bound-aux-graph}.
For every edge $uv$ in~$M'$, we add the arcs $(u_{\text{in}}, v_{\text{out}})$ and $(v_{\text{in}}, u_{\text{out}})$ to $\oH$.
Note that $v_{M, \text{out}}$ remains \exposed throughout the process, so we refer to the in-vertex $v_{M, \text{in}}$ simply as $v_M$.

\begin{figure}[htb]
    \begin{minipage}[b]{0.5\textwidth}
        \includegraphics[page=1]{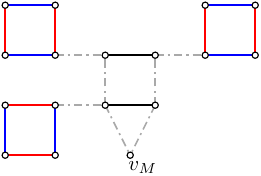}
        \subcaption{}
        \label{fig:linear-bound-aux-graph-1}
    \end{minipage}
    \begin{minipage}[b]{0.5\textwidth}
        \includegraphics[page=2]{figures/linear-bound-aux-graph.pdf}
       \subcaption{}
        \label{fig:linear-bound-aux-graph-2}
    \end{minipage}
    \caption{(a) A graph $G$ with the initial and target matching $M$ and $M'$, respectively, and (b)~the initialization of the auxiliary graph $\oH$.}
    \label{fig:linear-bound-aux-graph}
\end{figure}

Let $t$ be the number of cycles in $M \triangle M'$. In the following, we modify $\oH$ by adding arcs. To this end, we iteratively construct auxiliary odd matchings $M_0, \dots, M_t$ in $G$. 
Set~$M_0 := M$.
Suppose $M_i$, for some $i \in \{0, \dots, t-1\}$, is the last constructed matching. 
We now \emph{process} $M_i$ to construct $M_{i+1}$.
To this end, we choose an arbitrary cycle $K$ in $M_i \triangle M'$.
By~\cref{lem:alternating_paths}, $G$ contains an $M_i$-alternating path $P$ from $v_M$ to $K$.
Let $u_{C_i}$ be the first vertex on a cycle~$C_i$ in $M_i \triangle M'$ that we encounter when traversing $P$ starting from $v_M$.
We~now traverse $P$ from~$u_{C_i}$ towards $v_M$. 
Suppose that $e = ww'$ is the current edge on the path and assume that $w'$ is closer to $v_M$ on $P$ than $w$.
If $e \in M_i$, we simply move on to the next edge in $P$.
If $e \notin M_i$, we add the arc $(w_{\text{out}}, w'_{\text{in}})$ to~$\oH$.
If there is no directed path from $w'_{\text{in}}$ to $v_M$ in~$\oH$, we traverse the next edge on $P$.
Otherwise, we stop the traversal and define $M_{i+1}$ as the matching obtained from $M_i$ by replacing $M_i \cap C_i$ by $M' \cap C_i$.
One such iteration is visualized in~\cref{fig:linear-bound-aux-graph-iteration}.

\begin{figure}[htb]
    \begin{minipage}[b]{0.5\textwidth}
        \includegraphics[page=3]{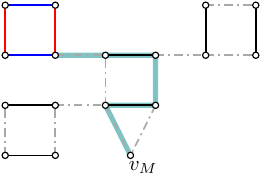}
        \subcaption{}
        \label{fig:linear-bound-aux-graph-iteration-1}
    \end{minipage}
    \begin{minipage}[b]{0.5\textwidth}
        \includegraphics[page=4]{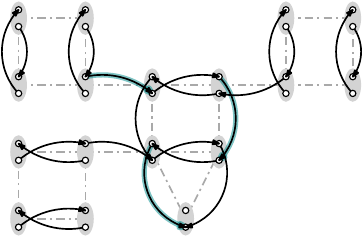}
       \subcaption{}
        \label{fig:linear-bound-aux-graph-iteration-2}
    \end{minipage}
    \caption{Iteration of the construction of $\oH$, at the moment of inclusion of the top-left cycle. \linebreak
    (a) The auxiliary matching $M_i$ with the highlighted alternating path connecting the cycle to $v_M$.
    (b) The corresponding state of $\oH$ and the highlighted corresponding arcs added to $\oH$.
    }
    \label{fig:linear-bound-aux-graph-iteration}
\end{figure}

We now collect some observations on the structure of $\oH$ and the matchings $M_0, \dots, M_t$.

\begin{restatable}[$\star$]{lemma}{lemObsAuxGraph}\label{lem:propertiesh}
\begin{itemize}
    \item We have $M_t = M'$ and for all $i \in \{0, \dots, t\}$, $v_M$ is \exposed in $M_i$.
    \item The arcs of $\oH$ always join an in-vertex and an out-vertex. In particular, $H$ is bipartite.
\item Each out-vertex in $\oH$ has at most one outgoing arc.
\end{itemize}
\end{restatable}
\medskip
In particular, \cref{lem:propertiesh} implies the following:

\begin{corollary}\label{cor:forest}
The graph $H$ is a forest. 
Further, let $\mathcal{C}$ be the connected component of $H$ containing $v_M$. 
In $\oH$, there is a directed path from every vertex in $\mathcal{C}$ to $v_M$. 
Moreover, for every cycle $K$ in $M \triangle M'$,  $\mathcal{C}$ contains the vertex $u_{K,\text{out}}$.
\end{corollary}

\begin{restatable}[$\star$]{lemma}{directedPathInAuxGraph}\label{lem:alternating}
If $(u_{\text{in}}, v_{\text{out}})$ is an arc in $\oH$, then $uv \in M'$. 
If $(u_{\text{out}},v_{\text{in}})$ is an arc in $\oH$, then $uv \in E(G) \setminus M'$.  
Hence, every directed path in $\oH$ defines an $M'$-alternating path in~$G$.  
\end{restatable}

\paragraph*{Description of the flipping procedure}

Using the graph $H$ above,
we now describe a flip sequence that transforms $M$ to $M'$.

Let $\mathcal{C}$ be the connected component of $H$ containing $v_M = v_{M, \text{in}}$.
We traverse~$\mathcal{C}$ in a depth-first search fashion, starting and ending at $v_M$. 
In parallel, we perform certain flip operations in $G$. We maintain the following property: whenever a vertex $v_{\text{in}}$ is visited in the traversal, the corresponding vertex $v$ is \exposed in $G$. 

Let $v_{\text{in}}$ be the currently visited vertex. Let $v_{1,\text{out}},\dots,v_{k,\text{out}}$ be all vertices of $\oH$ such that $(v_{i,\text{out}},v_{\text{in}})$ is an arc. Two cases can occur: 

\begin{enumerate}[(1)]
    \item All vertices $v_{1,\text{out}},\dots,v_{k,\text{out}}$ have already been visited. Then we backtrack. That is, if $v_{\text{in}}=v_M$, we terminate. Otherwise, let $w_{1,\text{out}}, w_{2,\text{in}}$ be the next vertices on the directed path from $v_{\text{in}}$ to $v_M$. 
    We will show that the edge $w_1 w_2$ is present in $G$. We flip it to $vw_1$, so that the \exposed vertex becomes $w_2$, and next visit $w_{2,\text{in}}$.
    \item There is an unvisited vertex $v_{i,\text{out}}$. In this case, we check whether $v_i$ is contained in a cycle~$C$ in $M\triangle M'$ that has not yet been switched, and if so, switch $C$.
    Subsequently, we consider the second neighbor $v_{\text{in}}'$ of $v_{i,\text{out}}$ in $H$. We will show that the edge $v'v_i$ is present in $G$. We flip it to $v_i v$, so that the \exposed vertex becomes $v'$. We mark $v_{i,\text{out}}$ as visited and then visit $v'_{\text{in}}$.
\end{enumerate}

\begin{restatable}[$\star$]{theorem}{thmCombinatorialDiameter}
    \label{thm:linear}
    Let $M$ and $M'$ be odd matchings of a graph $G$ such that $M$ can be reconfigured to $M'$. Using the above procedure, we can reconfigure $M$ to $M'$ using $\BigO(n)$ flips.
\end{restatable}

\subsection{Convex point sets}

We start with a refinement of the notion of happy edges in the case of convex point sets. 
\emph{Good happy edges} are defined by the following recursive definition:
	
	\begin{enumerate}[(1)]
		\item Happy edges on the convex hull are necessarily good happy edges.
		\item If a happy edge $e$ splits the convex point set into two parts such that one part contains the \exposed vertex (possibly along with some edges) and all edges in the other part are happy edges (implying that they all are good happy edges), then $e$ is a good happy~edge.
	\end{enumerate}
	\begin{figure}[ht]
		\centering
		\includegraphics{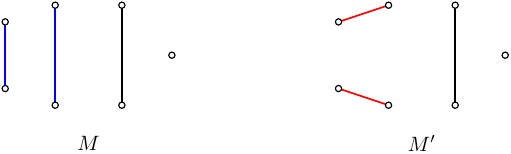}
		\caption{Example of a convex point set in which a happy edge needs to be flipped.}
		\label{fig:matchings_3}
	\end{figure}
    A happy edge that is not good is a \emph{bad} happy edge.
	\Cref{fig:matchings_3} provides an example in which a happy edge needs to be flipped. However, this edge is not a good happy edge.
    On the contrary, good happy edges are preserved in shortest flip sequences.
	
    \begin{restatable}[$\star$]{lemma}{Goodhap}
        \label{lem:goodhappy}
		No shortest flip sequence between any two odd matchings on a convex point set flips good happy edges.
    \end{restatable}

    Our recursive definition of good happy edges allows us to construct an equivalent instance without happy edges in linear time.

    \begin{restatable}[$\star$]{remark}{GOODHAPPY}
    \label{rem:goodhappy}
        Let $\matin$ and $\mattar$ be two matchings. We can remove good happy edges and their incident vertices in time that is linear in the number of vertices.
    \end{restatable}

    \begin{theorem}\label{thm:FPT}
		The flip distance $k$ between two odd matchings of the same convex point set~$S$ is fixed-parameter tractable in $k$.
	\end{theorem}
	
	\begin{proof}
		We remove all good happy edges using~\cref{rem:goodhappy}.
		As all remaining edges need to be flipped at least once, the instance is a ``no'' instance, if there are more than $k$ edges left. 
        In each step, there are at most $2k$ endpoints of matching edges. Adding an edge between the \exposed vertex and an endpoint determines which edge gets removed. 
        Thus, we need to check at most $2k$ possible flips per step and at most $(2k)^k$~sequences.
	\end{proof}

    \begin{remark}
        Actually, \cref{thm:FPT} shows a slightly stronger property: 
        The problem of determining whether the flip distance between two odd matchings is at most $k$ admits a polynomial size kernel with at most $k$ edges and at most $2k+1$ vertices. 
        That is, for every instance consisting of two matchings $M$ and $M'$ along with a parameter $k$, there exists a corresponding instance whose size is polynomial in $k$, such that this new instance is a \texttt{YES} instance if and only if the original instance is a \texttt{YES} instance.
    \end{remark}
	
	Consider the union of an initial matching $\matin$ and a target matching $\mattar$.
    With respect to this, let $A$ be the set of all good happy edges, $B$ the set of all bad happy edges, $C$ the set of edges in $\matin$ that lie on even alternating cycles, $c$ the number of such cycles, and $D$ the set of edges in $\matin$ that lie on the alternating path. 
    We have observed that all edges in $C$ and~$D$ are flipped at least once, and all edges in $B$ are flipped at least twice. 
    Additionally, each even alternating cycle requires one extra flip to place the \exposed vertex onto the cycle.
    The final flip in a cycle returns the \exposed vertex to the component it originally came from. 
    This yields a lower bound on the flip distance, which is tight when the union of $\matin$ and $\mattar$ is crossing-free.
    
    \begin{restatable}[$\star$]{theorem}{UP}
    \label{thm:up}
		Let $\matin$, $\mattar$, $A$, $B$, $C$, $c$, and $D$ be as defined above, and let $\matin\cup \mattar$ be crossing-free. 
        Then the flip distance from $\matin$ to $\mattar$ is $2\lvert B\rvert+\lvert C\rvert+c+\lvert D\rvert$.
    \end{restatable}

    We only give an intuition for the proof of \cref{thm:up}. The idea is to place the \exposed vertex once on every alternating cycle in $\matin\cup\mattar$ without making any unnecessary flips. In this step, it helps that in a convex point set if a vertex sees a connected component, it always sees two consecutive vertices. Therefore, we do not have to worry about parity constraints. Then, since the two matchings are together crossing-free, we can greedily switch every cycle in the optimal number of flips.

    \begin{corollary}\label{cor:radius}
        The radius of the flip graph of odd matchings for $n=2m+1$ points in convex point sets is at most $\frac{3m}{2}-\BigO(1)$. As a consequence, the diameter of the flip graph is at most $3m-\BigO(1) = \frac{3n}{2}-\BigO(1)$.
    \end{corollary}

    \begin{proof}
        For any given matching $M$ we show that we can flip $M$ into a given matching $M'$ that has all its edges on the convex hull. Observe that all happy edges of $M\cup M'$ are good happy edges, and that it contains at most $\frac{m}{2}$ alternating cycles and one alternating path. Plugging the parameters into~\cref{thm:up} we obtain the desired upper bound. 
        Since the diameter of a graph is bounded by twice its radius, we obtain an upper bound on the diameter of $3m-\BigO(1) = \frac{3n}{2}-\BigO(1)$.
 \end{proof}

    \begin{restatable}[$\star$]{remark}{lowerbound}
        \label{rem:radius-lower-bound}
        For any given odd matching $M$ there exists a matching $M'$ on $n=2m+1$ points in a convex point set such that $d(M,M') = \frac{3m}{2}-\BigO(1)$.
    \end{restatable}

With this result, we improve the previously known upper bound on the diameter from $4m - \BigO(1) = 2n - \BigO(1)$, as given in~\cite{oddmatchings}, to $3m - \BigO(1) = \frac{3n}{2} - \BigO(1)$. 
By also establishing a matching lower bound on the radius, we demonstrate that this upper bound is best possible among all approaches that proceed via a canonical intermediate structure. 
Any further improvement to the upper bound on the diameter will likely require more sophisticated techniques, potentially exploiting structural similarities between $\matin$ and $\mattar$.
    \section{Hardness results}
\label{sec:hardness}

In this section, we analyze the computational complexity of the problem in various settings, specifically for point sets in general position, grid graphs, and more broadly, planar graphs. 
We prove that deciding whether there exists a flip sequence of length at most~$k$ between two odd matchings is \NP-complete in each of these settings.

\subsection{\NP-completeness in point sets}
\label{sec:hardness_pointsets}

\begin{restatable}{theorem}{hardnessPointSet} \label{thm:general_hard}
    Let $\matin$ and $\mattar$ be two odd matchings on a set of $n$ points in the plane. 
    Deciding whether there is a flip sequence of length $k$ transforming $\matin$ into $\mattar$ is \NP-complete.
\end{restatable}

We reduce from the \NP-complete problem \textsc{Planar Monotone 3SAT}~\cite{planarmonotone3sat}.
Membership in \NP\ follows easily from the fact that a valid flip sequence of length at most $k$ serves as a certificate.
As the flip graph of odd matchings in this setting has diameter~$\mathcal{O}(n^2)$ as shown in~\cite{oddmatchings}, a valid flip sequence has polynomial length.

\subparagraph*{Planar Monotone 3SAT.} 
This is a variant of \textsc{3SAT} where each clause has at most three literals, either all positive or all negative. 
The variable-clause incidence graph is planar and must be embedded with variables on a horizontal line, positive clauses above, and negative clauses below. 
Moreover, every such instance can be represented in a rectilinear form~\cite{knuth1992problemcompatiblerepresentatives}, where each variable and clause corresponds to a rectangle, and their relationship is depicted by a vertical line segment connecting them; we refer to~\cref{fig:3sat} for an illustration.

\begin{figure}[ht]
	\centering
	\includegraphics{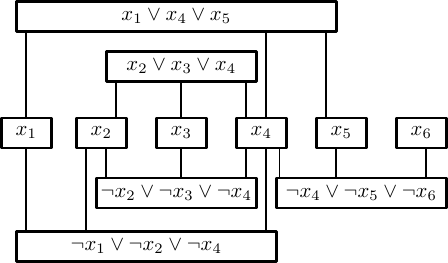}
	\caption{The clause-variable incidence graph of the Boolean formula $\Phi$ in conjunctive normal form with $\Phi = (x_1\vee x_4 \vee x_5)\wedge(x_2\vee x_3 \vee x_4)\wedge(\neg x_1\vee \neg x_2 \vee \neg x_4)\wedge(\neg x_2\vee \neg x_3 \vee \neg x_4)\wedge(\neg x_4\vee \neg x_5 \vee \neg x_6)$ as an instance of \textsc{Planar Monotone 3SAT}.}
	\label{fig:3sat}
\end{figure}

\subparagraph*{Variable gadget.} 
The variable gadget is shown in~\cref{fig:var}. 
While the black edge is happy, the blue and red edges are only contained in the initial and target configuration, respectively.
    
    \begin{figure}[htb]
		\centering
		\includegraphics{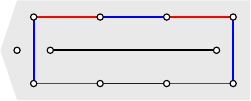}
		\caption{The variable gadget with the isolated vertex on the left.}
		\label{fig:var}
	\end{figure}

There are two optimal flip sequences that transform the variable gadget from its initial state to its target state. 
Two representative intermediate configurations are shown in~\cref{fig:var2}.
The direction of traversal encodes the truth value assigned to the variable.

    \begin{figure}[htb]
		\centering
		\includegraphics{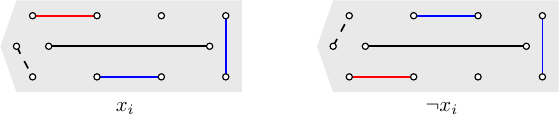}
		\caption{The variable gadget with a truth assignment: true on the left, false on the right.}
		\label{fig:var2}
	\end{figure}
	
\subparagraph*{Clause gadget.}
The clause gadget is depicted in~\cref{fig:clause}; note that the figure also includes three variable gadgets to illustrate the overall embedding.
As before, blue and red edges indicate the initial and target state, respectively. 
If a clause contains only positive literals, we place it above the corresponding variable gadgets, and if it contains only negative literals, we place it below.
Additionally, we introduce edges, shown in black, which are part of both matchings; they serve to obstruct visibility between the vertices of the variable and clause gadgets, permitting interaction only through a single designated pair.

\begin{figure}[ht]
	\centering
	\includegraphics{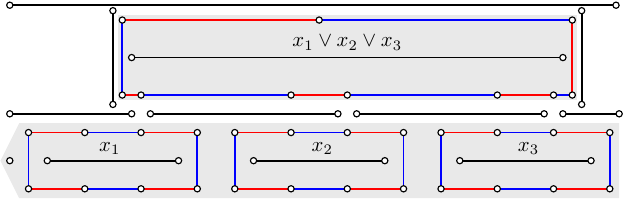}
	\caption{Clause gadget corresponding to $(x_1 \vee x_2 \vee x_3)$, with the associated variable gadgets included for clarity.}
	\label{fig:clause}
\end{figure}

\Cref{fig:clause2} illustrates how a valid truth assignment enables interaction with a clause gadget. 
During an optimal reconfiguration of a variable gadget, there is a point in the sequence where an \exposed vertex becomes visible to a vertex in the clause gadget if and only if the chosen sequence reflects a valid truth assignment. 
Moreover, when the clause gadget is reconfigured optimally, no \exposed vertex within the clause gadget will ever have visibility to a vertex in any other variable gadget.
\pagebreak
\subparagraph*{Rectilinear representation.} 
We complete the construction using the rectilinear representation of the \textsc{Planar Monotone 3SAT} instance. 
Each variable and clause rectangle is replaced by the corresponding gadget. 
Happy edges are added to block unintended visibility, while vertical segments in the rectilinear layout define visibility gaps between designated vertices of variable and clause gadgets.
We refer to~\cref{fig:3sat,fig:3satr} for illustrations.

\begin{figure}[htb]
	\centering
	\includegraphics{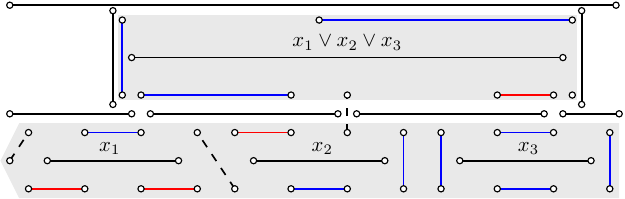}
	\caption{Clause gadget in an intermediate state. The clause $(x_1\vee x_2 \vee x_3)$ is satisfied because the variable $x_2$ is assigned \texttt{TRUE}.
    In contrast, $x_1$ is assigned \texttt{FALSE} and therefore does not contribute to reconfiguring the clause.}
	\label{fig:clause2}
\end{figure}

\begin{figure}[htb]
	\centering
	\includegraphics{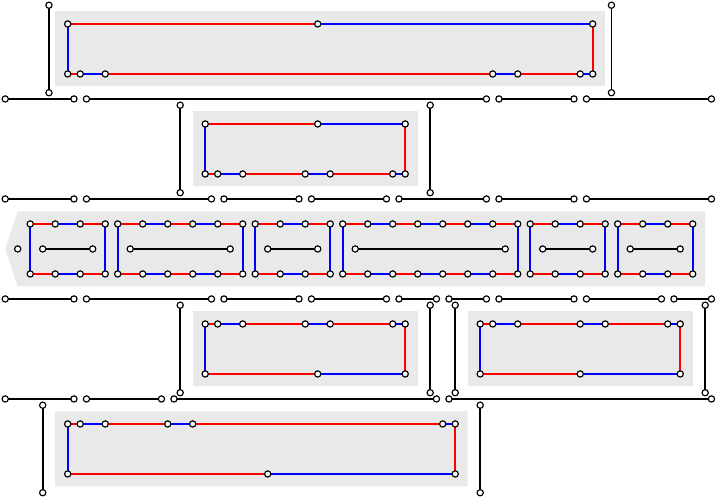}
	\caption{An instance of the odd matching reconfiguration problem, derived from the \textsc{Planar Monotone 3SAT} instance depicted in~\cref{fig:3sat}.}
	\label{fig:3satr}
\end{figure}

\begin{restatable}[$\star$]{proposition}{hardnessPointSets} \label{prop:general_hard}
	Let $\Phi$ be an instance of \textsc{Planar Monotone 3SAT} with $C$ clauses and $V$ variables. 
    There exists a flip sequence from $\matin$ to $\mattar$ (the corresponding odd matching instances) of length $\frac{1}{2}\lvert \matin \triangle \mattar \rvert+V+C$ flips if and only if $\Phi$ has a valid truth assignment.
\end{restatable}

\Cref{thm:general_hard} follows from \cref{prop:general_hard}. 
In our reduction, we construct point sets that are not in general position, as illustrated in the figures. 
To address this, \cref{lem:to_general} shows how to transform the point set into general position while preserving the essential structure.

\subsection{Generalizing to point sets in general position}
\label{subsec:general-point-sets}
This section is dedicated to show the following lemma.

\begin{lemma}\label{lem:to_general}
Let $\ps$ be a point set that can be embedded in an $h \times w$ grid.
Then there exists a mapping $\phi$ that maps each point of $\ps$ to a point on an $f(h,w) \times g(h,w)$ grid where $f$ and~$g$ are polynomial functions in $h$ and $w$ such that 
\begin{enumerate}
    \item for every three points $p, q, r$ of $\ps$, $\phi(p), \phi(q)$, and $\phi(r)$ are not collinear,
    \item if $p,q$, and $r$ are three non-collinear points of $\ps$, then the orientation of the triplet $(p,q,r)$ is the same as that of $\big(\phi(p), \phi(q), \phi(r)\big)$, and 
    \item for every two points $p$ and $q$ of $\ps$, $\phi(p)$ and $\phi(q)$ have different $x$- and $y$-coordinates.
\end{enumerate}
\end{lemma}

An $\varepsilon$-square is a square with side length $\varepsilon$. 
Our main idea is to place small $\varepsilon$-squares~$B_{\varepsilon}(v)$ centered at each vertex $v$.
These $\varepsilon$-squares are small enough such that for any triple of non-collinear points $p,q,r \in S$, a line $p'q'$ with $p' \in B_{\varepsilon}(p)$ and $q' \in B_{\varepsilon}(q)$ does not cross~$B_{\varepsilon}(r)$.
Further, we blow up our grid such that there are enough points inside each $\varepsilon$-square, and hence, there is at least one grid point that shares no $x$- or $y$- coordinate with any other point and that is not on a line spanned by two points of $S$.

\begin{restatable}[$\star$]{lemma}{distToGrid}
\label{lem:dist_to_grid}
 Let $\ell$ be a line through two points of an $h\times w$ square grid $G$.
 For every point $p$ of $G$ not on $\ell$, let $p^x$ and $p^y$ be the point on $\ell$ that shares the same $x$- and $y$-coordinate with $p$, respectively. 
 Then the distance from $p$ to each of $p^x$ and $p^y$ is at least $\frac{1}{2hw}$.
\end{restatable}

The lemma above is used to show the following lemma.

\begin{restatable}[$\star$]{lemma}{smallwiggle}    
\label{lem:small_wiggle}
 Let $p, q$ be two points of an $h \times w$ grid $G$. Let $p'$ and $q'$ be points in the $\varepsilon$-squares centered at $p$ and $q$, respectively.
 If $\varepsilon \leq \frac{1}{16w^3h^3}$, then for every grid point $r$ that does not lie on the line through $p$ and $q$, the orientation of the triplet $(p,q,r)$ is the same as $(p',q',r')$ for every point $r'$ in the $\varepsilon$-square centered at $r$.
\end{restatable}

We are now ready to prove \cref{lem:to_general}.

\begin{proof}[Proof of~\cref{lem:to_general}]
Let $\ps$ be a set of $n$ points, and denote them by $p_1, \dots, p_n$.
Further let $\varepsilon = \frac{1}{16h^3w^3}$. 
We construct the new grid $G'$ such that for every vertex $p$ of the point set, the $\varepsilon$-square $B_{\varepsilon}(p)$ centered at $p$ contains $(n^4 + n) \times (n^4 + n)$ grid points.
The resulting grid~$G'$ has now size $h\cdot 16h^3w^3 \cdot (n^4 + n) \times w \cdot 16h^3w^3 \cdot (n^4 + n)$ which is polynomial in $h$ and $w$.

We inductively define the mapping $\phi$.
We maintain that for $i = 1, \dots, n$, the set $\ps_i := \{\phi(p_1), \dots, \phi(p_i)\}$ satisfies the conditions (1), (2), and (3) of the lemma.

We define $\phi(p_1) = p_1$.
The set $\ps_1$ trivially satisfies (1)--(3).
Suppose we have constructed the set $\ps_i$ for some $i \in \{1, \dots, n-1\}$.
We now construct the set $\ps_{i+1}$.
Consider ${i \choose 2}$ lines spanned by any two points in $\ps_i$ and the $i$ vertical lines and $i$ horizontal lines passing through a point in $\ps_i$.
The former ${i \choose 2}$ lines intersect in less than $(i^2)^2 < n^4$ grid points in $B_{\varepsilon}(p_{i+1})$.
Hence, even when we avoid the $i$ vertical and $i$ horizontal lines, we can choose one point in $B_{\varepsilon}(p_{i+1})$ to be $\phi(p_{i+1})$ such that $\ps_{i+1}$ satisfies (1) and (3).
It remains to prove that $\ps_{i+1}$ also satisfies (2).
It is sufficient to consider a triplet $(p_a, p_b, p_{i+1})$ for $1 \leq a < b \leq i$.
By construction, $\phi(p_a) \in B_{\varepsilon}(p_a)$ and $\phi(p_b) \in B_{\varepsilon}(p_b)$.
Hence, by \cref{lem:small_wiggle}, the orientation of $(p_a, p_b, p_{i+1})$ is the same as that of $\big(\phi(p_a), \phi(p_b), \phi(p_{i+1})\big)$.
The lemma then follows.
\end{proof}
\pagebreak
\subsection{\NP-completeness in square grid graphs}
\label{sec:hardness_graphs}

A \emph{rectilinear Steiner tree} is a tree that connects a given set $K$ of points in the plane by only using horizontal and vertical line segments.
The \textsc{Rectilinear Steiner Tree Problem}~\cite{Steinertree} asks whether there exists a tree that connects a given set of vertices using only horizontal and vertical line segments, such that the total edge length does not exceed a specified bound.

The problem is known to be strongly \NP-hard; that is, it remains \NP-hard even when all vertices lie on a grid of polynomial size. 
Moreover, there always exists an optimal rectilinear Steiner tree contained within the Hanan grid~\cite{Hanan}, which is constructed by extending horizontal and vertical lines from each vertex in the input set. 
This result also implies that the grid can be arbitrarily refined if needed, without affecting the edges of an optimal solution.

We reduce from the \textsc{Rectilinear Steiner Tree Problem} to obtain the following.

\begin{theorem}\label{thm:grid_hard}
    Deciding whether there exists a flip sequence of length $\ell$ between two odd matchings in grid graphs is \NP-hard (and consequently in planar and general graphs as well).
\end{theorem}

\subparagraph*{High-level overview.}
Let $K$ be a set of $k$ points in an $n\times n$ grid.
We assume, without loss of generality, that $n>k$. 
We aim to construct two odd matchings on the grid graph such that any short flip sequence between them corresponds to a rectilinear Steiner tree of short total edge length.
Our reduction works in two phases: 
(1) By leveraging the structure of the Hanan grid, we restrict the input to point sets whose coordinates are all congruent to~$1\bmod 8n$, and (2) by carefully constructing two odd matchings on a sufficiently large grid graph, we ensure that short flip sequences between them correspond to short trees.

\subparagraph*{Construction.}
Let $v_{low}$ be a point in $K$ with the smallest $y$-coordinate (break ties arbitrarily). 
Without loss of generality, assume that its $y$-coordinate is $0$.
We embed the point set $K$ into an $8n\cdot(n + 4n^2)\times 8n\cdot(n + 4n^2)$ grid graph $G'$; in particular, each point $v_i=(x_i,y_i)\in K$ is mapped to the vertex $(8n\cdot x_i+1,\,8n\cdot y_i+1)$. 
We obtain a new set of points denoted by~$K'$.

\begin{restatable}[$\star$]{lemma}{HANAN}
\label{lem:hanan}
$K$ can be connected by an RST of length $\ell$ if and only if $K'$ can be connected by an RST of length $8n\ell$.
\end{restatable}

Next, we replace every vertex in $G'$ as follows: if a vertex belongs to $K$, we replace it by eight vertices arranged in an alternating cycle of four edges from $\matin$ and $\mattar$; otherwise, we replace it with eight vertices and four happy edges; see~\cref{fig:steiner_1} for an illustration. 

\begin{figure}[htb]
	\centering
	\includegraphics{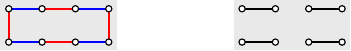}
	\caption{Vertex gadgets for vertices in $K$ (left), and not in~$K$ (right).}
	\label{fig:steiner_1}
\end{figure}
Then we add an additional row and an additional column to the bottom and to the right. We place the \exposed vertex of both $\matin$ and $\mattar$ below the third column of the gadget that corresponds to $v_{low}$. For the remaining points there is a unique way to connect them to an odd matching. The construction is illustrated in \cref{fig:steiner_2}. We remark that in order to keep the instance smaller, we do not use the refined version of the grid (i.e., $G'$) in the illustration, we also do not add the additional $4n^2$ vertices to the right and top of the relevant area.

\begin{figure}[ht]
	\centering
	\includegraphics{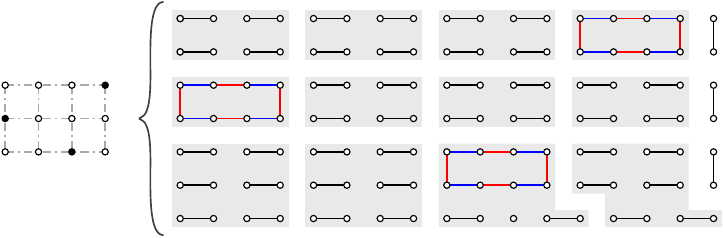}
	\caption{Reduction from rectilinear Steiner tree, massively scaled down for illustration.}
	\label{fig:steiner_2}
\end{figure}

\begin{restatable}[$\star$]{proposition}{hardnessFlipSequenceGrid}
\label{prop:grid-reconfiguration-length}
For any instance of RSTP with $k$ vertices on an $n\times n$ grid, there exists an RST of total edge length at most $\ell$ if and only if there exists a flip sequence for two corresponding odd matchings $\matin, \mattar$ transforming $\matin$ into $\mattar$ of length at most $4\cdot 8n\cdot \ell + 8k$.
\end{restatable}

With this, \Cref{thm:grid_hard} can be easily obtained from~\cref{prop:grid-reconfiguration-length}.
    \section{Conclusions and future work}
\label{sec:conclusions}

In this work, we investigated the reconfiguration of odd matchings and provided a complete characterization of when the associated flip graph is connected.
This complements a recent result that showed connectedness of the flip graph in the geometric setting~\cite{oddmatchings}.
We further showed that, in the combinatorial setting, each connected component of the flip graph has diameter linear in the cardinality of the matchings.
We established that the problem is \NP-complete, both in general graphs (in particular, even in rectangular grid graphs) and in geometric settings. 
For the geometric setting, we developed a framework to embed partial order types into general position. We believe that our framework can be applied in several similar proofs.
Lastly, for point sets in convex position, we provided a fixed-parameter tractable algorithm and improved the upper bound on the diameter of the flip graph.

\medskip
The following questions are related to our research. 
\smallskip
\begin{itemize}
\item Considering that planarity constraints make things harder: Can the linear bound for the diameter of the flip graph in the combinatorial setting be generalized to the geometric setting? Or can it be used to derive bounds on the diameter of the flip graph in the geometric setting?
\item What is the complexity of finding shortest flip sequences between odd matchings in convex point sets? Closing the gap between the upper bound and the lower bound of the diameter can be seen as a first promising direction.
\item Now that we established hardness of the flip distance problem: What can be said about fixed-parameter tractability and approximability of flip distances in both the combinatorial and the general geometric setting?
\end{itemize}

    \bibliography{references}    
    
    \appendix
    \section{Omitted details for~\cref{sec:flipgraph}}
\label{app:flipgraph}

\rectangularGridReconfigurable*

\begin{proof}
    If either dimension is equal to $1$, the grid graph is isomorphic to a path with an odd number of vertices, which can be easily reconfigured.
    Now assume both dimensions are at least $3$. 
    It is straightforward to verify that in a $3\times 3$ grid graph, removing any one of the four corner vertices or the center vertex results in a graph that still contains a Hamiltonian cycle.
    Since the $3\times 3$ grid is an induced subgraph of any larger grid with dimensions $w,h \geq 3$, we can shift this subgraph around the grid while preserving vertex parity.
    Doing so partitions the grid into smaller rectangular components, each with an even number of vertices, guaranteeing the existence of a Hamiltonian cycle in each.
    These local cycles can then be merged into a single global Hamiltonian cycle, demonstrating that for any vertex $v$ of the correct parity, that is, a vertex from the partition class containing exactly one more vertex than the other, the graph $G-v$ admits a Hamiltonian cycle.
\end{proof}

\theoremCharacterizationAlgorithm*

\begin{proof}
    We use as a subroutine an algorithm that returns a maximum cardinality matching algorithm of a given graph in time $\BigO(n^{2.5})$~\cite{matchingdeterministic}.
    In particular, this subroutine can be used to check if the graph has an odd matching or a perfect matching.
    
    We start by checking whether $G$ admits an odd matching, i.e., whether the flip graph is empty or not.
    Then for each edge in $G$, we check the conditions of~\cref{thm:reconfigurable}, as follows.
    Firstly, note that $e = uv$ lies in every odd matching of $G$ if and only if the graph $G - e$ contains no odd matching. Similarly, $e$ does not belong to any odd matching of $G$ if and only if $G - u - v$ has no odd matching.
    Hence, we can check \cref{thm:reconfigurable} by checking if at least one of $G-e$, $G-u-v$, $G-u$, or $G-v$ contains a matching of maximum size.
    Since there are $\BigO(n^2)$ edges, the overall runtime is $\BigO(n^{4.5})$.
\end{proof}
    \section{Omitted details for~\cref{sec:easy_setting}}
\label{app:easy_setting}

\lemObsAuxGraph*
\begin{proof}
For $i \in \{0, \dots, t\}$, $v_M$ is the \exposed vertex of $M_i$ as this holds for $M_0 = M$ and $M_{i+1}$ arises from $M_i$ by flipping the edges on an alternating cycle. 
Note that $M_{i+1} \triangle M'$ contains at least one cycle less than $M_i \triangle M'$, namely $C_i$.
Hence $M_t \triangle M'$ is empty, so $M_t = M'$. The second claim follows by construction. 
For the last, let $\oH_i$ be the graph $\oH$ at the point when we define $M_i$. Suppose that when processing $M_i$, we add an outgoing arc to $v_{\text{out}}$ for some $v \in V(G)$. By construction, there is a directed path from $v_{\text{out}}$ to $v_M$ in $\oH_{i+1}$, so we will never add another outgoing arc to $v_{\text{out}}$.
\end{proof}

\directedPathInAuxGraph*
\begin{proof}
The first assertion is by construction.
Now let $(w_{\text{out}}, w'_{\text{in}})$ be an arc that is added to $\oH$ while processing $M_i$. Suppose for the sake of contradiction that $ww' \in M'$.
By construction, $ww' \notin M_i$.
This implies that $ww' \in M_i \triangle M'$. 
As $M_i \triangle M'$ is a disjoint union of cycles,
$w$ and $w'$ lie on some cycle in $M_i \triangle M'$, which contradicts the choice of $u_{C_i}$. 
\end{proof}

\thmCombinatorialDiameter*
\begin{proof}
The symmetric difference $M \triangle M'$ is a union of cycles and a path $P$, whose endpoints are the \exposed vertices $v_M$ and $v_{M'}$. 
Flipping the edges along $P$ as before moves $v_M$ to $v_{M'}$. This requires $\BigO(n)$ flips. 
Hence, we can assume $v_M = v_{M'}$ and apply the above procedure.

\subparagraph*{Correctness.} 
We show that the operations performed in the algorithm are feasible.

Suppose that Case {\sffamily{\textbf{(2)}}} applies, i.e., we are visiting $v_{\text{in}}$, and there is an unvisited vertex $v_{i,\text{out}}$, with $v'_{\text{in}}$ its second neighbor. 
If $v_i$ belongs to a cycle $C$ in $M\triangle M'$ that was not yet switched, it is possible to switch $C$ since the current exposed vertex $v$ is adjacent to $v_{i}$.
Suppose that $v'v_i$ is not present in the current matching of $G$. 
Then $v',v_i$ are contained in the same cycle of $M\triangle M'$, which is then switched, so that $v'v_i$ is present afterwards.

For the correctness of the backtracking in {\sffamily{\textbf{(1)}}}, we check that the following property is maintained: when visiting $v_{\text{in}}$, the directed path from $v_{\text{in}}$ to $v_M$ corresponds to an alternating path in the current matching of $G$.
At the start of the algorithm the condition trivially holds since $v=v_M$.
Now let $v_{\text{in}}$ be the currently visited vertex, and assume that the above condition holds.
It is easily seen that switching a cycle does not affect the alternating path. Furthermore, performing an edge switch, either forwards or by backtracking, increases or decreases the length of this path by $2$, respectively, but maintains the condition. Thus it is maintained throughout, and thus the backtracking step is feasible. 

Every vertex in $\mathcal{C}$ is visited at least once, as the algorithm essentially performs depth-first search. 
Since every cycle in $M\triangle M'$ has at least one vertex in $\mathcal{C}$, it is eventually switched.
Furthermore, if an edge in $M'$ is flipped due to visiting a new vertex, it is restored when backtracking.
Thus, at the end of the algorithm the resulting odd matching on $G$ is~$M'$.

\subparagraph*{Number of flips.} The total number of edges contained in cycles in $M \triangle M'$ is linear in $n$. 
Switching a cycle of length $2k$ requires $k+1$ flips. Since the number of cycles is linear in $n$, the total number of flips required for all cycle switchings is $\BigO(n)$.
Finally, $H$ has $\BigO(n)$ vertices, and since it is a forest (\cref{cor:forest}), it also has $\BigO(n)$ edges.
In the DFS tour, each edge of $\mathcal{C}$ is traversed exactly twice. 
Hence, the overall number of flips performed is $\BigO(n)$.
\end{proof}

\Goodhap*
    \begin{proof}
		It is sufficient to show the statement for good happy edges on the convex hull. For other good happy edges, iteratively consider new instances in which all good happy edges from the convex hull and the vertices incident to them have been removed and repeat the proof. By our recursive definition it is easy to see that every good happy edge will be on the convex hull eventually.
		
		For proving the case that $e$ is a convex hull edge, we introduce a normalization technique which describes a more refined version of the normalization that was used in \cite{articlematchings} for a similar proof involving crossing-free perfect matchings.
		
		For a plane odd matching $\bar{M}$, we define the \emph{normalization} $N(\bar{M})$ with respect to the edge $e=uv$ as follows:
		
		\begin{itemize}
			\item If $e \in \bar{M}$, then $N(\bar{M}) = \bar{M}$
			\item If $e$ is incident to the vertex that is not matched in $\bar{M}$ (without loss of generality $v$), let~$e'$ be the edge that is incident to $u$ in $\bar{M}$. Then $N(\bar{M}) = (M\setminus \{e'\}) \cup e$.
			\item If $u$ and $v$ are incident to edges $e'$ and $e''$ in $
            \bar{M}$ respectively let $e^\ast$ be the edge that connects the two endpoints of $e'$ and $e''$ that are different from $u$ and $v$. Then ${N(\bar{M}) = (\bar{M}\setminus\{e',e''\}) \cup \{e,e^\ast\}}$.
		\end{itemize}
		
		Because $e$ is an edge of the convex hull, $N(\bar{M})$ will always be a plane odd matching. 
        Further, if two matchings $\bar{M}$ and $\bar{M}'$ only differ by a single flip $N(\bar{M})$ and $N(\bar{M}')$ will either coincide or only differ by one flip. 
        This can be seen in the commutative diagrams in~\cref{fig:mat}.

        \begin{figure}[htb]
            \begin{minipage}[b]{0.5\textwidth}
                \includegraphics[page=1]{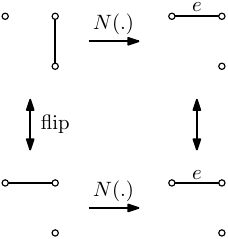}
                \subcaption{}
                \label{fig:mat-1}
            \end{minipage}
            \begin{minipage}[b]{0.5\textwidth}
                \includegraphics[page=2]{figures/normalize__new.pdf}
                \subcaption{}
                \label{fig:mat-2}
            \end{minipage}
            \caption{Commutative diagrams that illustrate that if two matchings are adjacent in the flip graph, so are their normalizations (a) $e$ is added or removed (b) an edge incident to a vertex of $e$ is added or removed.
            }
            \label{fig:mat}
        \end{figure}
		
		For a given flip sequence $\matin = M_0$, $M_1,\dots,M_{k-1}$, $M_{k} = \mattar$ for which $e$ is removed at some point (without loss of generality from $M_0$ to $M_1$) and later added back (without loss of generality from $M_{k-1}$ to $M_k$) observe that $N(M_0) = M_0$, $N(M_1) = N(M_0)$, $N(M_k) = M_k$, and $N(M_{k-1}) = N(M_k)$. 
        This yields that the flip sequence $N(M_0$), $N(M_1),\dots,N(M_{k-1})$, $N(M_{k})$ is again a flip sequence from $\matin$ to $\mattar$ that is a least two flips shorter than the original one.
	\end{proof}

\GOODHAPPY*

\begin{proof}
        We start by assigning numbers $0,\ldots,n-1$ to the vertices of the convex point set in clockwise order, beginning at the \exposed vertex of $\matin$. For both given matchings, we create an array of size $n$, say $I[0,\ldots,n-1]$ for the initial matching, and $T[0,\ldots,n-1]$ for the target matching. For each array an entry at index $i$ contains the index of the vertex that is matched to the vertex with index $i$ in the matching or an indication that the vertex is \exposed. Further, for every $i$ we save the two indices of the vertices that lie (index wise) next to $i$ on the current subset of the convex hull. At the start all neighboring vertices are just $i-1$ and $i+1$ taken modulo $n$. The arrays $I[]$ and $T[]$ can be initialized in linear time by simply traversing the matchings and storing the indices.
    	Our recursive definition of good happy edges allows us to remove good happy edges using a stack $\Sigma$. We add the vertices from $0$ to $n-1$ in that order to $\Sigma$. Every time when the top two vertices form a good happy edge, we remove both of them from $\Sigma$. The top two vertices form a good happy edge if they are connected to one another in both matchings. This can be checked by accessing the respective entries for the vertices in $I$ and $T$ and checking the matched vertex. This step can be executed in constant time for every pair of vertices. Further, we update the entries in the arrays that correspond to neighboring vertices such that now the vertices before and after the good happy edge lie next to one another. Since every vertex is inserted in $\Sigma$ and removed from $\Sigma$ at most once, all good happy edges can be removed in linear time. This leaves us with a set $S'$ without good happy edges.
\end{proof}

\UP*
\begin{proof}
    	First, omit all good happy edges and their incident vertices to obtain a reduced set of vertices $V'$. For all remaining connected components in $E = \matin \cup \mattar$ of Types B-D we construct a plane graph $G$ where every vertex describes a component and two vertices are connected by an edge if their corresponding components have vertices that appear consecutively on the convex hull of $V'$. Let $T$ be a spanning tree of $G$. We root $T$ at the component of type D and traverse $T$ starting from there. All the leaves (except possibly the root) correspond to components of Type C, because if they were happy edges, they would have to be good happy edges. See~\cref{fig:tree} for an illustration. By flipping the components in the order that is determined by $T$, we make sure to reconfigure every component in the amount of flips given in the above case distinction without additional overhead.

        \begin{figure}[ht]
			\centering
			\includegraphics{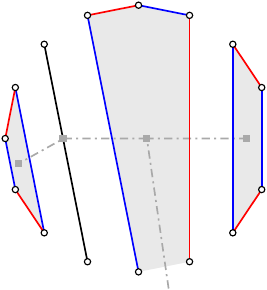}
			\caption{Dual graph of the union of two matchings.}
			\label{fig:tree}
		\end{figure}
		
		We modify each component $H$ according to the following case distinction, assuming that the \exposed vertex $v$ has already been placed such that some vertex of $H$ adjacent to it on the convex hull:
        
        \begin{enumerate}[(1)]
			\item If $H$ is a leaf vertex of $T$, then it corresponds to an even cycle with $k$ edges from $\matin$. We~place the \exposed vertex and switch $H$ in $k+1$ flips. The \exposed vertex is again $v$.
			\item If $H$ is an internal vertex of $T$ that corresponds to a bad happy edge $h$, we add an edge~$f$ from $v$ to the vertex of $h$ that is next to $v$ on the convex hull $v$. In the process, we create a new \exposed vertex $v'$ at one the other end of $H$. We use $v'$ to add edges that place the \exposed vertex in connected components that are children of $H$ in $T$. After fixing performing all relevant flips in the child components, we add $h$ back making $v$ the \exposed vertex again. We performed two flips in $H$.
			\item $H$ is an internal vertex of $T$ that corresponds to a cycle we switch the cycle almost the same way as in Case {\sffamily{\textbf{(1)}}}, with one difference. Whenever we create a new \exposed vertex~$v'$ that lies next to a component that corresponds to a subtree of $H$ in $T$ on the convex hull, we use $v'$ to place the \exposed vertex on connected components that correspond to children of $H$ in $T$. A total of $k+1$ flips is performed in $H$.
			\item $H$ is the root vertex of $T$ that corresponds to the path of Type D with $k$ edges from the initial matching of $T$. One by one, we replace edges from the initial matching with edges from the target matching. Whenever we create a new \exposed vertex $v'$ that lies next to a component that corresponds to a subtree of $H$ in $T$ on the convex hull, we use $v'$ to perform all the required steps in those components before we continue with traversing~$H$. In total $k$ flips are charged towards Case {\sffamily{\textbf{(4)}}}.
		\end{enumerate}
			
			By traversing all components in the order that was determined by $T$, we make sure that all cycles of Type C as well as the path of Type D is fixed. Summing up over all the flips required in every case, we get a flip sequence of length $2\lvert B\rvert+\lvert C\rvert+c+\lvert D\rvert$.
\end{proof}

\lowerbound*
\begin{proof}
        For the lower bound we show that for any given odd matching we can construct a second matching that has flip distance at least $\frac{3m}{2}-\BigO(1)$.

        Let $M$ be an arbitrary odd matching on a convex point set. We show that there exists a matching $M'$ that is compatible with $M$ such that the union of $M$ contains $\big\lfloor\frac{\lvert M \rvert}{2}\big\rfloor$ cycles with $4$ edges and possible one happy edge. The required lower bound then follows from~\cref{thm:up}.

        First, remove the \exposed vertex from $M$. Now, if all edges of $M$ are on the convex hull, construct $M'$ such that it connects two arbitrary consecutive convex hull edges via the convex hull edge between them and add the remaining edge to close the $4$-cycle. Repeat the process until only no (no more steps to do) or one edge (add the edge to $M'$ as well) is left.

        We are also done, if $M$ only has two edges, because in this setting all edges are on the convex hull (recall that we removed the \exposed vertex).

        Now let $M$ contain non convex hull edge $e$. This edge splits the point set into two parts. If one of the parts contains only one edge, said edge has to be a convex hull edge $h$ and we connect $e$ and $h$ in $M'$ to form a non-crossing 4-cycle. If both parts contain more than one edge, we split $M$ along $e$ and apply our approaches on both parts separately. We have to add $e$ to one of the two parts. We do this in such a way that not both resulting parts contain an odd amount of edges.
\end{proof}
    \section{Omitted details for~\cref{sec:hardness}}
\label{app:hardness}
We now provide the omitted technical details underlying the hardness results.

\subsection{\NP-hardness for point sets}
\hardnessPointSets*

\begin{proof}
    Given two odd matchings $\matin$ and $\mattar$ on a set of points, we aim to show that there is a flip sequence of length $k = \frac{1}{2} \lvert\matin \triangle \mattar\rvert + V + C$ transforming $\matin$ into $\mattar$ if and only if there is a satisfying assignment for a corresponding Boolean formula with $V$ many variables and $C$ many clauses as an instance of \textsc{Planar Monotone 3SAT}.

    Given an instance of \textsc{Planar Monotone 3SAT}, we construct an instance of the odd matching reconfiguration problem as follows:
    For each variable, we introduce a variable gadget, arranged horizontally in a single row as illustrated in~\cref{fig:var}.
    For every clause consisting solely of positive literals, we place a clause gadget above the variable row, guaranteeing that the interaction between vertices from the clause and vertices from the respective variables is restricted to a single designated pair, as shown in~\cref{fig:clause}.
    For clauses containing only negated literals, we place a clause gadget below the variable row in an analogous manner.
    Finally, we add an additional \exposed vertex to the left of the leftmost variable gadget to complete the construction.
    
    We take a close look at every at $E=M\cup M'$. 
    $E$ consists of the following connected components: 
    the alternating path that connects the two \exposed vertices reduces to a single vertex as $v_M=v_{M'}$; 
    $C$ many alternating cycles, one corresponding to each clause gadget; 
    $V$ many alternating cycles, one corresponding to each variable gadget;
    all edges introduced to block visibilities are happy edges, i.e., they are contained in both matchings.
    From the observations in \cref{sec:prelim-obs}, we derive a lower bound on the flip distance of $\frac{1}{2} \lvert\matin \triangle \mattar\rvert + V + C$. 
    
    We now argue that this lower bound is tight exactly when the Boolean formula is satisfiable.
    In particular, flips cannot be charged to components consisting solely of a single happy edge.
    Moreover, each gadget must be transformed using the minimum number of flips required for its reconfiguration.

    \begin{claim}\label{clm:point-set-hardness-assignment}
        If $\Phi$ is satisfiable, then there is a flip sequence of length exactly $\frac{1}{2} \lvert\matin \triangle \mattar\rvert + V + C$ reconfiguring $\matin$ into $\mattar$.
    \end{claim}
    
	\begin{claimproof}
        If $\Phi$ has a satisfying truth assignment, traverse each variable gadget in an order that reflects the truth value assigned to its corresponding variable.
        For every clause that is satisfied by some variable $x_i$, make a detour---as illustrated in~\cref{fig:clause2}---to reconfigure the corresponding clause gadget. 
        The resulting flip sequence has length $\frac{1}{2} \lvert\matin \triangle \mattar\rvert + V + C$.
    \end{claimproof}

    \begin{claim}\label{clm:point-set-hardness-sequence}
        If there is a flip sequence of length exactly $\frac{1}{2} \lvert\matin \triangle \mattar\rvert + V + C$ reconfiguring $\matin$ into $\mattar$, then $\Phi$ is satisfiable.
    \end{claim}

    \begin{claimproof}
	Now assume that there exists a flip sequence of length exactly $\frac{1}{2} \lvert\matin \triangle \mattar\rvert + V + C$. 
    Then no vertex other than vertices in the variable gadgets and the clause gadgets can be involved in any flips; otherwise, we would have removed and re-added a happy edge, which would incur at least two additional flips. 
    Furthermore, each variable gadget is traversed in one of two optimal ways: 
    if the gadget corresponding to variable~$x_i$ is traversed as shown on the left of~\cref{fig:var2}, we assign the value \texttt{TRUE}; otherwise, we assign \texttt{FALSE}.
    Because all clause gadgets are reconfigured optimally, the traversal of variable gadgets must ensure that, for each clause gadget, there exists at least one \exposed vertex in a variable gadget that connects to it.
    This guarantees that the traversal of the variable gadgets corresponds to a truth assignment that satisfies all clauses of $\Phi$ simultaneously, and therefore $\Phi$.
    \end{claimproof}

    This concludes the proof of~\cref{prop:general_hard}.
\end{proof}

\subsection{Generalizing to point sets in general position}

\distToGrid*

\begin{proof}

First note if $\ell$ is a vertical line then any point of $G$ that is not on $\ell$ has at least distance $1$ to $\ell$.
If $\ell$ is not a vertical line, any line $\ell$ spanned by two points of $G$ is of the form $y=\frac{a}{b}x+q$ where $-h\leq a \leq h$ and $0 < b \leq w$ and $q$ is a point of $G$.
This means that for integer values of $x$, the corresponding $y$-value is an integer or its distance to any integer is at least $\frac{1}{w}$. If the corresponding $y$-value is an integer, then it is a point of $G\cap \ell$ and has distance one to any point of $G$ that is not on $\ell$.
Similarly, if $y$ is an integer, the corresponding $x$-value is an integer or its distance to any integer is at least~$\frac{1}{h}$.
Since $h,w \geq 1$, the lemma then follows.
\end{proof}

\smallwiggle*

\begin{proof}
Let $p=(x_p,y_p), q=(x_q,y_q), p'=(x_{p'},y_{p'}), q=(x_{q'},y_{q'})$.
Note that $|x_p-x_{p'}|, |x_q-x_{q'}|, |y_p-y_{p'}|, |y_q-y_{q'}|  < \varepsilon$. 
Let $d_x=x_p-x_q, d_y = y_p-y_q$.
Let $r$ be a grid point that does not lie on the line $pq$ and $r'$ be a point inside the $\varepsilon$-square centered at $r$. Note that the distance from $r$ to $r'$ is at most $\varepsilon<\frac{1}{16hw}$.

If $d_x = 0$, $pq$ is a vertical line.
Note that by the assumption of $\varepsilon$, $p'q'$ then cannot be a horizontal line.
$|x_{p'}-x_{q'}|\leq \varepsilon$ and $|y_{p'}-y_{q'}|\geq |y_p-y_q| - \varepsilon \geq 1- \varepsilon > \frac{1}{2}$.
So the slope of~$p'q'$ is greater than $\frac{1}{2\varepsilon}$ or smaller than $-\frac{1}{2\varepsilon}$. 
This implies that within the area of the grid, the difference between the $x$-coordinates of any point $a'$ on $p'q'$ and that of $p'$ is at most $2\varepsilon h$.
Hence, the distance from $a'$ to $pq$ is at most $2\varepsilon h + \frac{\varepsilon}{2} < \frac{1}{4hw}$, since $\varepsilon < \frac{1}{16h^3w^3}$.
By \cref{lem:dist_to_grid}, the distance from $r$ to $pq$ is at least $\frac{1}{2hw} > \frac{1}{4hw}$. Therefore the distance from $r'$ to $p'q'$ is at least $\frac{1}{2hw}-\frac{1}{4hw} - \frac{1}{16hw}=\frac{3}{16hw}$.

This implies that $r'$ is on the right of $p'q'$ if and only if it is on the right of $p'q'$.
Hence, the lemma holds in this case.

Now consider the case when $d_x \neq 0$.
Then by the assumption of $\varepsilon$, $p'q'$ cannot be a vertical line.
We call the \emph{$y$-distance} from a point $s$ to the line $pq$ is the distance between $s$ and the point $s^x$ on $pq$ that shares the same $x$-coordinate with $s$.

\begin{claim}
\label{claim:y_distance}
Let $s$ be a point on $p'q'$ in the area of the grid $G$.
Then the $y$-distance from $s$ to the line $pq$ is at most $\frac{1}{4hw}$.
\end{claim}
\begin{claimproof}
The slope of the line $pq$ is $\frac{d_y}{d_x}$.
Note that $d_x-2\varepsilon\leq x_{p'}-x_{q'}\leq d_x + 2\varepsilon$.
Furthermore, $d_y-2\varepsilon\leq y_{p'}-y_{q'}\leq d_y + 2\varepsilon$.
Hence, $\frac{d_y-2\varepsilon}{d_x+2\varepsilon}\leq \frac{y_{p'}-y_{q'}}{x_{p'}-x_{q'}} \leq \frac{d_y+2\varepsilon}{d_x-2\varepsilon}$.

Note that $\left\lvert \frac{d_y}{d_x}-\frac{y_{p'}-y_{q'}}{x_{p'}-x_{q'}}\right\lvert \leq \frac{d_y+2\varepsilon}{d_x-2\varepsilon} - \frac{d_y}{d_x}$ since $\frac{d_y}{d_x}-\frac{d_y-2\varepsilon}{d_x+2\varepsilon}\leq \frac{d_y+2\varepsilon}{d_x-2\varepsilon} - \frac{d_y}{d_x}$.  
Further,

\begin{align*}
    \frac{d_y+2\varepsilon}{d_x-2\varepsilon} - \frac{d_y}{d_x} &= \frac{d_y+2\varepsilon}{d_x (1-\frac{2\varepsilon}{d_x})} - \frac{d_y}{d_x} \\
    &\leq \frac{d_y(1-\frac{2\varepsilon}{w}) + d_y\frac{2\varepsilon}{w}+2\varepsilon}{d_x (1-\frac{2\varepsilon}{w})} - \frac{d_y}{d_x} \\
    &\leq 4\varepsilon \left(\frac{d_y}{w} + 1\right) \\
    &< 4\varepsilon h
\end{align*}

Hence, the slope changes by at most $4\varepsilon h$.

Since $p'$ is in $B_{\varepsilon}(p)$ the distance between $p$ and $p'$ is at most $2\varepsilon$.
Therefore, inside the grid $G$, for any point $\bar{r}$ on the line $p'q'$, its distance to the point $\bar{r}^x$ on $pq$ that shares the $x$-coordinate with $\bar{r}$ is at most $2\varepsilon + 4\varepsilon h w < \frac{1}{4hw}$.
\end{claimproof}

By \cref{lem:dist_to_grid}, the $y$-distance from $r$ to $pq$ is at least $\frac{1}{2hw}$.
Combined with \cref{claim:y_distance} above and that the distance from $r$ to $r'$ is at most $\varepsilon$, we deduce that $r$ is above $pq$ if and only if $r'$ is above $p'q'$.
The lemma then follows.
\end{proof}

\subsection{NP-hardness for grid graphs}

\HANAN*

\begin{proof}
    Assume $K$ admits an RST of length at most $\ell$. Subdividing every line segment of length $1$ into $8n$ and embedding it in the finer grid yields an RST that connects all vertices of $K'$ and has length at most $8n\ell$.

    Now let $K'$ admit an RST of length at most $8n\ell$. Then, again, there exits a Steiner tree in the Hanan grid of length at most $8n\ell$. Each line segment then has a length that is a multiple of $8n$. If we remove all vertices that are not $1~mod~8n$ and replacing vertical or horizontal line segments with $8n$ points on them with a single vertical or horizontal line segment gives us an RST that connects $K$ and has at most $\ell$ line segments.
\end{proof}

\hardnessFlipSequenceGrid*
\begin{proof}
     First we transform $K$ into $K'$ according to \cref{lem:hanan}.

    \begin{claim}\label{clm:grid-hardness-sequence}
        If there exists an RST of length at most $8n\ell$ in the refined grid $G'$, then there exists a flip sequence of length at most $4\cdot(8n)\cdot \ell + 8k$
    \end{claim}

    \begin{claimproof}
    If there exists a solution to RSTP with length at most $8n\cdot\ell$, then there also exists a solution that lies entirely in the Hanan grid and, in particular, has at most $k$ corners.

    Take said tree in the Hanan grid and traverse it in depth first search order. Let each traversed edge in the tree correspond to an alternating path between two vertex in $K'$ gadgets that walk along vertices in the refined grid graph $G'$ along the edge. Traversing a vertex not in $K'$ gadget horizontally or vertically takes two flips to traverse them once and another two flips to go back, so a total of $4$ flips. If a vertex corresponds to a corner, it might need more flips when being traversed twice, but never more than $6$. When a vertex in $K$ gadget is traversed for the last time, we replace all the remaining edges from $\matin$ in it with edges in $\mattar$.

    In total we invest up to $4\cdot (8n) \cdot \ell$ flips for vertical and horizontal traversal of vertex gadgets, $2k$ additional flips for possible corners and $5k$ to turn the vertex in $K$ gadgets from their initial state to their target state. This adds up to less than $4\cdot (8n) \cdot \ell +8k$ flips.
    \end{claimproof}

    \begin{claim}\label{clm:grid-hardness-tree}
        If there exists a flip sequence of length at most $4\cdot(8n)\cdot \ell + 8k$, then there exists an RST of length at most $8n\ell$ in the refined grid $G'$ and consequently an RST of length $\ell$ in the original grid.
    \end{claim}
    
    \begin{claimproof}
   If there exists a flip sequence of length less than $4\cdot (8n) \cdot \ell +8 k$, we draw a graph in the grid graph of the refined instance of RSTP (with the point set $K'$). The graph contains every vertex whose corresponding gadget had an edge flipped in the flip sequence. We add an edge between two vertices if and only if during the flip sequence, there exists an edge between the two corresponding vertex gadgets. The constructed graph connects all vertices in $K'$. First, we take a spanning tree of the constructed graph to eliminate cycles. The total length of the spanning tree is at most the length of the spanning graph. Second, we exhaustively remove vertices that are leaf vertices and do not lie in $K'$. Again, the resulting tree $T$ only gets shorter.

    Last, we construct the final solution $T'$ to the refined instance of RSTP as follows.
    Let~$S$ be the set of vertices of $T$ that are either in $K'$ or have degree at least three.
    Since the leaves are in $K'$ and there are $k$ points in $K'$, there are at most $2k$ vertices in $S$.
    For any two vertices $u$ and $v$ in~$S$ such that the path $P$ between them in $T$ does not contain any other vertices in~$S$, we add to $T'$ one vertical line segment $L_V$ and one horizontal line segment $L_H$ to connect the two vertices.
    It is easy to see that the total length $|T'|$ of $T'$ is at most that of~$T$.
    We now consider the flip sequence corresponding to the path $P$.
    Without loss of generality, assume that the leftmost bottommost point in the vertex gadget $R_u$ of $u$ is~$(0,0)$ and that in the vertex gadget $R_v$ of $v$ is $(4|L_H|, 2|L_V|)$.
    Note that the flip sequence corresponds to an $M$-alternating path.
    We say an edge on this path is \emph{relevant} if it is in~$M$ and not in $R_u$ nor in $R_v$.
    By the choice of $u$ and $v$, a relevant edge must be happy.
    In other words, their incident vertices of a relevant edge are of the form $(2\alpha,\beta)$ and $(2\alpha+1,\beta)$; we label such an edge by the vertex of the lower coordinate (i.e, $(2\alpha,\beta)$).
    Observe that for two consecutive relevant edges, either only the $x$-coordinates differ by two or only the $y$-coordinates differ by one.
    This implies that the relevant edges is at least $2|L_H| + 2|L_V| -1$.
    Since these edges are happy, we need to traverse at least twice.
    This implies that the number of flips associated with the pair $u, v$ is $4|L_H| + 4|L_V| - 2$.
    Summing over all possible pairs $u,v$ and taking into account the additional $5k$ flips to flip all the alternating cycles, the length of entire flip sequence has to be at least $5k + 4|T'| -\sum_{(u,v)} 2$.
    Since $T$ is a tree, the number of pairs $u,v$ is at most $|S| \leq 2k$.
    Combining the preceding two sentences with the fact that the length of the flip sequence is at most $4\cdot (8n) \cdot \ell +8 k$, we have
    \[
        |T'| \leq (8n) \cdot \ell + \frac{7}{4}k.
    \]
    
    Since the refined instance has an optimal solution on the Hanan grid, the optimal value has to be a multiple of $8n$.
    Combining the above with the fact that the optimal value is at most $|T'|$, we conclude that the optimal value is at most $\ell$.
    \end{claimproof}

    This concludes the proof of~\cref{prop:grid-reconfiguration-length}.
\end{proof}
\end{document}